\documentstyle[amsfonts,11pt]{article}
 
\title{Domain Wall Equations, Hessian of Superpotential, and Bogomol'nyi Bounds}
\author{Shouxin Chen\footnote{Email address: chensx@henu.edu.cn}\\Institute of Contemporary Mathematics\\School of Mathematics and Statistics\\Henan University\\
Kaifeng, Henan 475004, PR China\\\\
Yisong Yang\footnote{Email address: yisongyang@nyu.edu}\\Department of Mathematics\\Tandon School
of Engineering\\New York University\\Brooklyn, New York 11201, U. S. A}
\date{}
\newcommand{\bfR}{{\Bbb R}}

\newcommand{\sech}{\mbox{ sech}}

\newtheorem{oldtheorem}{Theorem}[section]
\newtheorem{oldassertion}[oldtheorem]{Assertion}
\newtheorem{oldproposition}[oldtheorem]{Proposition}

\newtheorem{oldlemma}[oldtheorem]{Lemma}
\newtheorem{olddefinition}[oldtheorem]{Definition}
\newtheorem{oldclaim}[oldtheorem]{Claim}
\newtheorem{oldcorollary}[oldtheorem]{Corollary}

\newenvironment{theorem}{\begin{oldtheorem}$\!\!\!${\bf.}}{\end{oldtheorem}}

\newenvironment{proposition}{\begin{oldproposition}$\!\!\!${\bf.}}
{\end{oldproposition}}
\newenvironment{lemma}{\begin{oldlemma}$\!\!\!${\bf.}}{\end{oldlemma}}

\newbox\qedbox
\newenvironment{proof}{\smallskip\noindent{\bf Proof.}\hskip \labelsep}%
                        {\hfill\penalty10000\copy\qedbox\par\medskip}

\newcommand{\dd}{\mbox{d}}
\newcommand{\ee}{\end{equation}}
\newcommand{\be}{\begin{equation}}\newcommand{\bea}{\begin{eqnarray}}
\newcommand{\eea}{\end{eqnarray}}
\newcommand{\e}{\mbox{e}}
\newcommand{\pa}{\partial}
\newcommand{\vep}{\varepsilon}

\newcommand{\nn}{\nonumber}

\newcommand{\lm}{\lambda}

\begin{document}
\maketitle
\begin{abstract} An important question concerning the classical solutions of
the equations of motion arising in quantum field theories at the BPS critical coupling is whether all finite-energy solutions
are necessarily BPS. In this paper we present a study of this basic question in the context of the domain wall equations
whose potential is induced from a superpotential so that the ground states  are the critical points of
the superpotential. We prove that the definiteness of the Hessian of the superpotential suffices to ensure that all
finite-energy domain-wall solutions are BPS. We give several examples to show that such a BPS property may fail such that
non-BPS solutions exist when
the Hessian of the superpotential is indefinite.
\end{abstract}
\section{Introduction}

Domain walls, vortices, monopoles, and instantons are classical solutions of various equations of motion in quantum field
theory describing particle-like behavior in interaction dynamics in one, two, three, and four spatial dimensions, 
respectively \cite{Wein}. Due to the complexity of these equations  it is difficult to obtain a full understanding of
their solutions in general settings. Fortunately, at certain critical coupling limits, enormous insight into the solutions may be obtained
from investigating the so-called BPS solutions, after the pioneering works of Bogomol'nyi \cite{B} and 
Prasad and Sommerfield \cite{PS}. Mathematically, at such critical limits, often referred to as the BPS limits, the 
original second-order
equations of motion permit an elegant reduction into some first-order equations, called the BPS equations, whose
solutions are automatically the energy minimizers of the models, and hence, the most physically relevant field configurations.
The minimum energy values achieved by the BPS solutions are called the BPS or Bogomol'nyi bounds and are often
expressed in terms of the topological charges of the models.
Physically, the BPS limits correspond to the situations that
the Higgs scalar and gauge boson masses are equal in
the Abelian Higgs model \cite{Hart}, the Higgs doublet and $Z$ vector boson masses are equal
in the electroweak theory \cite{AO}, both for vortices, the Higgs potential vanishes but symmetry is
spontaneously broken at an arbitrary level for monopoles \cite{JT}, and only gauge fields are present
for instantons \cite{A,R,Nash}. There has also been some study on
why BPS bounds exist in quantum field theory \cite{HS}. For more recent developments of the application of the BPS
reduction in supersymmetric field theory, see \cite{Kon-survey,ShY2,ShY,Tong} for surveys.

At the BPS coupling then an important question arises: Are the original second-order equations of motion equivalent to their
BPS-reduced first-order equations? In other words, are all finite-energy critical points of the field-theoretical energy functional the solutions of the BPS equations, and hence, attain the BPS bounds? For the Abelian Higgs vortices, Taubes 
proved \cite{JT,T1} that the answer is yes, and for the non-Abelian Yang--Mills--Higgs monopoles, he established \cite{T2}
that the answer is no such that there are nonminimal solutions of the Yang--Mills--Higgs equations in the BPS coupling which are
not solutions to the BPS equations. For the Yang--Mills instantons, Sibner, Sibner, and Uhlenbeck \cite{SSU} gave a no
answer to the question and proved the existence of a nonminimal solution when the instanton charge is zero, and Bor \cite{Bor}, Parker \cite{Parker}, and Sadun and Segert \cite{SS1,SS2,SS3} obtained the existence of nonminimal solutions
when the instanton charge is non-zero and non-unit, by exploring an equivariant structure in the geometric construction of the Yang--Mills fields. Much earlier, inspired by the work of Taubes \cite{T2} based on a Morse theory consideration, Manton \cite{Manton} and Klinkhamer
and Manton \cite{KM} investigated the possible existence of saddle point solutions, also referred to as sphalerons whose
energy gives rise to the height of a barrier for tunneling between two vacuum states of distinct topological charges, in
the Weinberg--Salam electroweak theory, by a demonstration that the field configuration space of the bosonic sector 
possesses a non-contractible loop. In \cite{FH} Forg\'{a}cs and Horv\'{a}th presented a series of examples of field-theoretical models in one, two, and three spatial dimensions that allow non-contractible loops in their configuration spaces. 
Hence these models may be candidates for the occurrence of saddle-point solutions and host barriers to topological
vacuum 
tunneling as in the electroweak theory \cite{KM,Manton}. These studies prompt a systematic investigation of the
existence of non-BPS solutions at the BPS limits, referred to here as the BPS problem.

Due to the complicated structures of various quantum field theory models, it will be difficult to obtain a thorough understanding of the BPS problem in its fully general setting. In this paper,  as a starting point, we consider a
general domain wall
model where the field configuration is an $n$-component real scalar field. As the name suggests, a domain wall is a dimensionally
reduced field configuration that is embedded into, and links two distinct domain phases realized as
ground states of,  the full field-theoretical model. Well-known classical examples of domain walls include the solutions of the sine--Gordon
equations describing a domain transition in terms of the magnetization orientation angle
in the Landau--Lifshitz theory of magnetism and the solutions of the Ginzburg--Landau equations
connecting the normal and superconducting phases so that the sign of its energy, called the surface energy, classifies
superconductivity. In modern physics, domain walls find applications in a wide range of subjects such as 
phase transitions in supersymmetric
field theory \cite{AT,CQR,SY} and supergravity theory \cite{CS}, monopole annihilation \cite{V1,V2}, and QCD confinement problem \cite{ABSY,ShY}. In our problem, the scalar field $u$ consists of $n$ real-valued scalar components, $u=(u_1,\dots,u_n)$,
and
the potential  $V(u)$ is given in terms of a superpotential $W(u)$ \cite{1,Alon,2,3,BLW,GT} such that $V(u)=\frac12|\nabla W(u)|^2$. Thus the minima (zeros) of $V$, which are the ground states of the model, are the critical points of $W$.
In this paper we show that the BPS problem is related to the definiteness of the Hessian of $W$ at its critical points
when these points are designated as the domain phases the domain wall solutions separate and connect.

An outline of the rest of the paper is as follows. In Section 2 we review the general multiple-scalar-field domain wall model,
the equations of motion, and the BPS equations,
where the potential (density) is induced from a superpotential. We then state and prove our main theorem that the
definiteness of the Hessian of the superpotential at at least one ground state (a phase domain) suffices to ensure that all finite-energy
domain wall solutions are BPS. In Sections 3--6 we present a series of examples to illustrate the applicability and limitation
of our main theorem. In particular, in Sections 4 and 5, we show by examples that when the definiteness of the Hessian fails there exist finite-energy domain
wall solutions which are not BPS. In Section 7 we conclude the paper.

\section{Domain wall equations and superpotential}
\setcounter{equation}{0}

Let $u=(u_1,u_2,\cdots,u_n)$ be a scalar field which is an $n$ real-component function over the Minkowski spacetime $\bfR^{1,1}$ of
signature $(+-)$ and use $\mu,\nu=0,1$ to denote the temporal and spatial coordinate indices
with $t=x^0$ and $x=x^1$ being the temporal and spatial coordinates, respectively. The domain wall
action density governing $u$ reads
\be\label{L}
{\cal L}=\frac12\sum_{i=1}^n\pa_\mu u_i\pa^\mu u_i-V(u),
\ee
where the potential energy density $V$ is given by
\be\label{x1.2}
V(u)=\frac12\sum_{i=1}^n W^2_{u_i}(u),\quad W_{u_i}=\frac{\pa W}{\pa u_i},
\ee
for a real-valued generating function $W$ \cite{1,Alon,2,3,BLW,GT}, commonly referred to as the superpotential in the supersymmetric field theory
formalism \cite{Alon,Beh,Dvali,Dvali2}, to be specified later. The Euler--Lagrange equations of (\ref{L}) are
\be\label{EL}
\ddot{u}_i-u_i''\equiv \pa_\mu\pa^\mu u_i=-\sum_{j=1}^n W_{u_i u_j} W_{u_j},\quad i=1,\dots,n,
\ee
which are $n$ coupled semilinear wave equations.
The energy is
\be\label{x1.0}
E(u)=\int^\infty_{-\infty}\left\{\frac12\sum_{i=1}^n \left(\dot{u}_i^2+{u'_i}^2\right)+V(u)\right\}\,\dd x,
\ee
which is conserved in view of (\ref{EL}).

We are interested in static solutions of (\ref{EL}). The energy of such a solution is
\be\label{x1.1}
E(u)=\int^\infty_{-\infty}\left\{\frac12\sum_{i=1}^n {u_i'}^2+V(u)\right\}\,\dd x.
\ee

In the static limit  the equations (\ref{EL}) become 
\be\label{x1.3}
u_i''=\sum_{j=1}^n W_{u_i u_j} W_{u_j},\quad i=1,2,\cdots,n,
\ee
which are the Euler--Lagrange equations of (\ref{x1.1}).

Let $u^1$ and $u^2$ be two zeros of $V$ or critical points of $W$, representing two distinct domain phases or ground states
of the model. We are interested in the domain-wall solutions of (\ref{x1.3}) which link the domain phases $u^1$ and $u^2$:
$
u(-\infty)=u^1, u(\infty)=u^2.
$

\subsection{First-order equations}

In view of the boundary condition $u(-\infty)=u^1, u(\infty)=u^2$, at the two spatial infinities, we can rewrite (\ref{x1.1}) as
\bea\label{x2.1}
E(u)&=&\int_{-\infty}^\infty \frac12\sum_{i=1}^n (u_i'\pm W_{u_i})^2\,\dd x\mp\int_{-\infty}^\infty \sum_{i=1}^n W_{u_i} u_i'\,\dd x\nn\\
&\geq&\left|W(u^2)-W(u^1)\right|,
\eea
according to
\be\label{x2.2}
W(u^2)-W(u^1)=\mp|W(u^2)-W(u^1)|.
\ee
Hence, the lower bound or the BPS bound on the right-hand side of (\ref{x2.1}) is attained if and only if $u$ satisfies the first-order BPS equations
\be\label{x2.3}
u_i'\pm W_{u_i}=0,\quad i=1,2,\cdots,n.
\ee

It is easy to check that (\ref{x2.3}) implies (\ref{x1.3}) regardless of the boundary condition. 

In this work we will identify some natural conditions that ensure  that (\ref{x1.3}) implies (\ref{x2.3}) as well so that
any finite-energy critical points of (\ref{x1.1}) must be absolute energy minimizers to achieve the BPS bound.

\subsection{A solution to the BPS problem}

To proceed, we begin by unveiling some natural conditions to be imposed on $W$. Recall that we are to show that a solution $u$ of (\ref{x1.3}) subject to the boundary condition $u(-\infty)=u^1, u(\infty)=u^2$
also fulfills the first-order system (\ref{x2.3}). To be specific, let us consider a test case when $W(u^1)>W(u^2)$.
 Thus we are to study (\ref{x2.3}) with the upper sign, i.e.,
\be\label{x3.7}
u'=-\nabla_u W,
\ee
{and that $u^1$ and $u^2$ are two equilibria of (\ref{x3.7}). Since we may view the independent variable $x$ dynamically (as a `time' variable), we see that the solution $u$ represents
an orbit connecting $u^1$ and $u^2$ as `time' evolves from $-\infty$ to $\infty$. Consequently, as the equilibrium points of (\ref{x3.7}), it would be natural to assume that $u^1$ is unstable but
$u^2$ is stable. This observation motivates the following condition
on the Hessian matrix $(W_{u_i u_j})$ at the two asymptotic states:
 \be\label{x3.8}
(W_{u_i u_j}(u^1))\mbox{ is negative definite;} \quad (W_{u_i u_j}(u^2))\mbox{ is positive definite. }
\ee
The desired solution of (\ref{x3.7}) then may well be viewed as representing a rolling particle that travels down from the local maximum $u^1$ to
the local minimum $u^2$ of the superpotential $W$.

Below is our main theorem aiming at a solution of the BPS problem of our interest.

\begin{theorem}\label{thm}
Consider the field-theoretical Lagrangian density (\ref{L}) governing an $n$ real-component scalar field $u$ so that
its static configuration
is of the energy (\ref{x1.1}) where the potential  $V$ is given by (\ref{x1.2}) with the
generating function or superpotential $W$ and let $u^1, u^2\in\bfR^n$ be two critical points of $W$
which are of course the zeros of $V$ serving as two ground states of the model. Moreover assume that
the Hessian $(W_{u_i u_j}(u))$ is either positive or negative definite at $u^1$ or $u^2$.
Then the Euler--Lagrange equations (\ref{x1.3}) and the BPS equations (\ref{x2.3}) are equivalent 
for solutions satisfying the boundary condition
\be\label{BC}
u(-\infty)=u^1,\quad u(\infty)=u^2,
\ee
which realizes a domain-wall type phase transition between the two ground states or phase domains.
Thus any finite-energy critical point of (\ref{x1.1}) satisfying the boundary condition
(\ref{BC})  must be an absolute energy minimizer
with the explicit BPS minimum energy value
\be\label{x3.20}
E=\left|W(u^1)-W(u^2)\right|.
\ee
Furthermore, if the Hessian $(W_{u_i u_j}(u))$ fails to be definite at both $u^1$ and $u^2$, then there are examples
that some finite-energy solutions of (\ref{x1.3}) are not the solutions of the BPS equations (\ref{x2.3}), subject to
the boundary condition (\ref{BC}), so that those solutions do not attain the BPS bound (\ref{x3.20}).
\end{theorem}

\begin{proof}
Let $u$ be a solution of (\ref{x1.3}) satisfying (\ref{BC}).  We establish the conclusion of the theorem when the Hessian
of $W$ is either negative or positive definite at $u^1$. In Sections 4 and 5, we shall present some examples
when the Hessian of $W$ is definite at neither $u^1$ nor $u^2$, as supplemental illustrations, to show in such a situation
that there exist non-BPS solutions.

First assume that $(W_{u_i u_j}(u^1))$ is negative definite. Hinted by the observation made preceding Theorem \ref{thm}, we aim to show that $u$ solves (\ref{x2.3}) with the upper sign. For this purpose, we set
\be\label{x3.2}
P_i=u_i'+W_{u_i},\quad i=1,2,\cdots,n,
\ee
and we are to establish the result
\be\label{x3.3}
P_i \equiv0, \quad i=1,2,\cdots,n.
\ee

To proceed, we differentiate (\ref{x3.2}) to get
\be\label{x3.4}
P_i'=u_i''+\sum_{j=1}^n W_{u_i u_j} u_j', \quad i=1,2,\cdots,n.
\ee
Thus,
we see in view of (\ref{x1.3}) that there holds
\be\label{x3.5}
P'_i=
\sum_{j=1}^n W_{u_i u_j} P_j,\quad i=1,2,\cdots,n.
\ee

In order to establish (\ref{x3.3}), we need only to show that there is a point $x_0\in\bfR$ such that 
\be\label{x3.6}
P_i(x_0)=0,\quad i=1,2,\cdots,n,
\ee
because then we can apply the uniqueness theorem for the initial value problems of ordinary differential equations to the system (\ref{x3.5}) with unknowns $P_1,P_2,\cdots,P_n$ subject
to the initial condition (\ref{x3.6}) to infer (\ref{x3.3}).

We now show that (\ref{x3.6}) must be valid for some $x_0\in\bfR$. In fact, setting
\be \label{xx3.7}
Q=|P|^2 =\sum_{i=1}^n P_i^2,
\ee
we derive from (\ref{x3.5}) that
\bea\label{x3.10}
Q'(x)&=&2\sum_{i,j=1}^n P_i(x) W_{u_i u_j} (u(x))P_j(x)\nn\\
&\leq&-\lm Q(x),\quad x<x_0\equiv -a,
\eea
where $\lm>0$ is a constant and $a>0$ is sufficiently large. Here we have used the boundary condition (\ref{BC})
and the assumption that $(W_{u_i u_j}(u^1))$ is negative definite. From (\ref{x3.10}), we obtain
\be\label{x3.11}
Q(x)\geq Q(x_0) \e^{\lm(x_0-x)},\quad -\infty<x<x_0.
\ee
If $Q(x_0)\neq 0$, then (\ref{x3.11}) implies that $Q(x)\to\infty$ as $x\to-\infty$. However, applying the finite-energy condition to (\ref{x1.1}), we have
\be\label{x3.12}
\liminf_{x\to-\infty} \sum_{i=1}^n (u_i'(x))^2 =0;
\ee
applying the boundary condition (\ref{BC}), we have $(\nabla_u W(u))(x)\to (\nabla_u W)(u^1)={\bf 0}$ as $x\to-\infty$. As a consequence, we have
\be\label{x3.13}
\liminf_{x\to-\infty} Q(x)=0,
\ee
which is a contradiction. Thus $Q(x_0)=0$ and (\ref{x3.6}) is valid. Hence (\ref{x3.3}) is proved.

We next assume that $ (W_{u_i u_j}(u^1))$ is positive definite. 

Let $u$ be a solution of (\ref{x1.3}) satisfying (\ref{BC}). Then, since (\ref{x1.3}) is invariant under the change of independent variable, $x\mapsto -x$, we see that $u=u(-x)$ is a solution of (\ref{x1.3}) satisfying the reversed boundary
condition
\be
u(-\infty)=u^2,\quad u(\infty)=u^1.
\ee
Now set $P_i$ by (\ref{x3.2}) and $Q$ by (\ref{x3.7}). We get as before
\bea
Q'(x)&=&2\sum_{i,j=1}^n P_i(x) W_{u_i u_j} (u(x))P_j(x)\nn\\
&\geq&\lm Q(x),\quad x>x_0,\label{x3.14}
\eea
where $x_0>0$ is sufficiently large and $\lm>0$ is a constant. If $Q(x_0)>0$ then (\ref{x3.14}) gives us $Q(x)\geq
Q(x_0)\e^{\lm(x-x_0)}$ so that $Q(x)\to\infty$ as $x\to\infty$ which contradicts again with the finite-energy property.
Hence $Q(x)\equiv0$ as before.

Returning to the original independent variable, we see that the equations
\be
u_i'-W_{u_i}=0,\quad i=1,2,\dots,n,
\ee
are now fulfilled, which belong to the second branch of the BPS equations, that is, (\ref{x2.3}) with the lower sign.
\end{proof}

\section{Chern--Simons domain walls}
\setcounter{equation}{0}

In \cite{Lee} Lee {\em et al} formulated a Maxwell Chern--Simons gauge field theory which generalizes the classical
Abelian Higgs theory \cite{JT} with an added neutral scalar field. In \cite{Bolog} Bolognesi and Gudnason considered
a class of phase transition scenarios realized by domain walls through a dimensional reduction procedure. In this section,
we apply Theorem \ref{thm} to gain a mathematical understanding of such domain walls.

Following \cite{Bolog}, the tension or energy of the Chern--Simons domain wall is given by
\be\label{xx4.1}
E=\int_{-\infty}^\infty\left\{\phi'^2 +\frac12 N'^2+U(\phi,N)\right\}\,\dd x,
\ee
where $\phi$ and $N$ are two real-valued scalar fields with the interaction potential 
\be\label{4.2}
U(\phi,N)=\frac12\left(e\phi^2+\kappa N-e v^2\right)^2+e^2 N^2\phi^2,
\ee
in which $e,\kappa,v>0$ are coupling constants. To put this model into the formalism here, we set
\be
\phi=\frac1{\sqrt{2}}\psi.
\ee
Thus (\ref{xx4.1}) and (\ref{4.2}) assume their normalized forms
\be\label{4.4}
E=\int_{-\infty}^\infty\left\{\frac12\psi'^2 +\frac12 N'^2+V(\psi,N)\right\}\,\dd x,
\ee
where
\be\label{4.5}
V(\psi,N)=\frac12\left(\frac e2\psi^2+\kappa N-e v^2\right)^2+\frac12e^2 N^2\psi^2
=\frac12\left(W_\psi^2+W_N^2\right),
\ee
where the superpotential $W$ is given by
\be\label{4.6}
W(\psi,N)=\frac e2 N\psi^2+\frac1{2\kappa}(\kappa N-ev^2)^2.
\ee
The Euler--Lagrange equations of (\ref{4.4}) are
\bea
\psi''&=&e^2 N^2\psi+e\left(\frac e2\psi^2+\kappa N-ev^2\right)\psi,\label{4.7}\\
N''&=&e^2\psi^2 N+\kappa\left(\frac e2\psi^2+\kappa N-ev^2\right).\label{4.8}
\eea
It is clear that
there are exactly three domain phases in terms of $u=(\psi,N)$ realized as critical points of $W$:
\be\label{4.9}
u^1=\left(-\sqrt{2} v,0\right),\quad u^2=\left(\sqrt{2} v,0\right),\quad u^3=\left(0,\frac{ev^2}\kappa\right).
\ee

On the other hand, the Hessian matrix of the function (\ref{4.6}) is
\be
H(u)=\left(\begin{array}{cc}W_{\psi\psi}&W_{\psi N}\\W_{\psi N}& W_{NN}\end{array}\right)=\left(\begin{array}{cc}eN&e\psi\\
e\psi&\kappa\end{array}\right),
\ee
which is indefinite at $u^1$ and $u^2$ and positive definite at $u^3$. 
Thus in the context of the domain wall solutions linking
the phases between $u^1$ or $ u^2$ and $u^3$, the equations (\ref{4.7})--(\ref{4.8}) are equivalent to the BPS equations
\cite{Bolog}:
\be\label{y3.11}
\left(\psi',N'\right)=\mp\left(W_\psi,W_N\right)=\mp\left(eN\psi,\frac e2\psi^2+\kappa N-e v^2\right),
\ee
whose existence problem has been settled in \cite{ZL}. Following our earlier discussion
in Section 2, we see that the 
energy  (\ref{4.4}) of a finite-energy domain wall  realizing the phase transition between $u^1$ or $u^2$ and $u^3$, must be the BPS energy
\be
T_{\mbox{wall}}=E_{\mbox{BPS}}=\left|W\left(\pm\sqrt{2}v,0\right)-W\left(0,\,\frac{ev^2}\kappa\right)\right|=\frac{e^2 v^4}{2\kappa},
\ee
as obtained in \cite{Bolog}.

It is interesting to know whether in the context of domain wall solutions linking the domains $u^1$ and $u^2$ the equations
(\ref{4.7})--(\ref{4.8}) are equivalent to (\ref{y3.11}). It is easily seen that this is out of question since the system
(\ref{y3.11}) has no such solution. In fact, if $(\psi, N)$ is a solution, then there is some $x_0\in(-\infty,\infty)$ such that
$\psi(x_0)=0$. Using this as the initial condition in the equation $\psi'=\mp e N(x)\psi$ and applying the uniqueness theorem
for the initial value problem of ordinary differential equations, we get $\psi\equiv0$, which is a contradiction.

\section{A two scalar field model}
\setcounter{equation}{0}

We now consider the domain wall model given by the energy
\be\label{4.1}
E(\varphi,\chi)=\int_{-\infty}^\infty\left\{\frac12 \varphi'^2+\frac12\chi'^2+V(\varphi,\chi)\right\}\,\dd x,
\ee
and studied in \cite{1,Alon,2,3,de} governing two real scalar fields $\varphi$ and $\chi$ for which
\bea\label{x4.7}
V(\varphi,\chi)&=&\frac12(1-\varphi^2)^2+\frac12 r^2\chi^4-r\chi^2+r(1+2r)\varphi^2\chi^2
=\frac12W^2_\varphi+\frac12 W_\chi^2,\\
W(\varphi,\chi)&=&\varphi-\frac13\varphi^3 -r\varphi\chi^2,
\eea
where $r\neq0$ is a constant.  The equations of motion are
\bea
\frac12\varphi''&=&(\varphi^2+r[1+2r]\chi^2-1)\varphi,\label{y3.4}\\
\frac12\chi''&=&r(r\chi^2-1+[1+2r]\varphi^2)\chi.\label{y3.5}
\eea

There are two cases of interest regarding the ground states.
\begin{enumerate}
\item[(i)] $r<0$. In this situation the potential $V$ has two zeros in terms of $u=(\varphi,\chi)$:
\be\label{x4.8}
u^1=(-1,0),\quad u^2=(1,0),
\ee

\item[(ii)] $r>0$. Now in terms of $u=(\varphi,\chi)$ the potential $V$ has four zeros:
\be
u^1=(-1,0),\quad u^2=(1,0),\quad u^3=\left(0,\frac1{\sqrt{r}}\right),\quad u^4=\left(0,-\frac1{\sqrt{r}}\right).
\ee
\end{enumerate}

The Hessian matrix of the generating function $W$ is
\be
H(u)=\left(\begin{array}{cc} W_{\varphi\varphi}&W_{\varphi\chi}\\W_{\varphi\chi}&W_{\chi\chi}\end{array}\right)
=-2\left(\begin{array}{cc}\varphi&r\chi\\ r\chi& r\phi\end{array}\right),
\ee
which is indefinite when $r<0$ and Theorem \ref{thm} is not applicable. On the other hand, when $r>0$, we see that
$
H(u^1)
$ is positive definite, $H(u^2)$ is negative definite, and $H(u^3)$ and $H(u^4)$ are indefinite.
Therefore, applying Theorem \ref{thm},  we see that the second-order equations of motion (\ref{y3.4}) and (\ref{y3.5})
are equivalent to the first-order BPS equations,
\be\label{x4.9}
(\varphi',\chi')=\pm(\varphi^2+r\chi^2-1,2r\varphi\chi),\label{x4.10}
\ee
in the context of finite-energy solutions for which either $u^1$ or $u^2$ is an asymptotic state as $x\to-\infty$ or $x\to\infty$.

For $0<r<\frac12$, the  solution of (\ref{x4.9})  satisfying $(\varphi,\chi)\to u^1$ as $x\to-\infty$ and $(\varphi,\chi)\to u^2$ as
$x\to\infty$ has been found \cite{1,2,3} to be given by
\be\label{x4.11}
\varphi(x)=\tanh (2r [x-x_0]),\quad \chi(x)=\sqrt{\frac1r-2}\sech(2r [x-x_0]),\quad x_0\in\bfR,
\ee
through an integration. Using Theorem \ref{thm} we see that (\ref{x4.11}) gives us all the solutions
of the original system (\ref{y3.4})--(\ref{y3.5}) satisfying the boundary condition $u(-\infty)=u^1$ and
$u(\infty)=u^2$. In particular the energy of such a solution must be the minimum BPS energy
\be\label{x4.12}
E=W(u^2)-W({u^1})=\frac43.
\ee.

We now aim at obtaining the solutions of (\ref{y3.4})--(\ref{y3.5}) with $r>0$ which link the ground states $u^3$ and $u^4$. First note that $W(u^3)=W(u^4)=0$.  So the BPS bound vanishes, which becomes unattainable. Alternatively,
it is clear that (\ref{x4.9}) allows no solution to make transition between $u^3$ and $u^4$. In fact, if there is such a
solution, then $\chi$ interpolates between $-\frac1{\sqrt{r}}$ and $\frac1{\sqrt{r}}$. Thus there is some $x_0\in\bfR$
such that $\chi(x_0)=0$. Since $\chi$ satisfies the equation $\chi'=\pm 2r\varphi\chi$, we see that $\chi\equiv0$ in
view of the uniqueness theorem for the initial value problem of ordinary differential equations, which is a contradiction.

Thus we have to look for solutions of the full system of the second-order equations (\ref{y3.4})--(\ref{y3.5}) in order
to be able to realize a phase transition between $u^3$ and $u^4$.
However, it may be hard to get such solutions due to the complicated structure of the equations. Fortunately we can
use the ansatz
$\varphi=0$ to reduce the equations into the single equation
\be\label{x3.12}
\chi''=2r(r\chi^2-1)\chi,
\ee
with the simplified energy
\be
E(0,\chi)=\frac12\int_{-\infty}^\infty\left(\chi'^2+(r\chi^2-1)^2\right)\,\dd x,
\ee
and the associated superpotential $W(\chi)=\frac r3\chi^3-\chi$, which spells out our desired asymptotics $\chi^2(\pm\infty)=\frac1r$. Thus we have the lower bound
\bea\label{y3.14}
E(0,\chi)&=&\frac12\int_{-\infty}^\infty \left(\chi'\pm (r\chi^2-1)\right)^2\dd x\mp\int_{-\infty}^\infty \chi'(r\chi^2-1)\,\dd x\nn\\
&\geq&\left|W\left(\frac1{\sqrt{r}}\right)-W\left(-\frac1{\sqrt{r}}\right)\right|=\frac4{3\sqrt{r}},
\eea
which is attained when $\chi$ solves the BPS equation
\be\label{x3.15}
\chi'\pm (r\chi^2-1)=0.
\ee
Since $W''\left(\pm\frac1{\sqrt{r}}\right)=\pm 2\sqrt{r}\neq0$, we see that (\ref{x3.12}) is equivalent to (\ref{x3.15}),
which is a classical result, whose solutions are given by the formulas
\be\label{y3.16}
\chi(x)=\pm\frac1{\sqrt{r}}\tanh \sqrt{r}(x-x_0).
\ee

Thus we have obtained a family of solutions of the coupled equations (\ref{y3.4}) and (\ref{y3.5}), linking
the ground states $u^3$ and $u^4$ where the Hessian of the superpotential fails to be definite, which are not solutions of the
BPS equations (\ref{x4.9}).

Comparing (\ref{x4.12}) with (\ref{y3.14}), we see that, when $r>1$, the energy carried by the non-BPS solution
$\varphi=0$ and $\chi$ is as given in (\ref{y3.16}) linking $u^3$ and $u^4$ where the superpotential is indefinite assumes
a lower value than that of the BPS solution linking $u^1$ and $u^2$ where the superpotential is definite. This is a rather
unexpected result.

\medskip

As another application, we consider the three-scalar field domain wall model \cite{BLW} defined by the static energy
density
\be
{\cal E}=\frac12{\varphi'}^2+\frac12{\chi'}^2+\frac12{\rho'}^2+V(\varphi,\chi,\rho),
\ee
where
\be
V(\varphi,\chi,\rho)=\frac12(1-\varphi^2)^2+2r^2\varphi^2\chi^2+\frac12r^2 (\chi^2+\rho^2)^2+r(\varphi^2-1)
(\chi^2+\rho^2)+2r^2 (\varphi-s)^2\rho^2,
\ee
with $r>0$ and $s\in (-1,0)\cup(0,1)$ being two coupling parameters. The Euler--Lagrange equations are
\bea
\frac12\varphi''&=&\left(\varphi^2-1+r[2r+1]\chi^2+r\rho^2\right)\varphi+2r^2\rho^2(\varphi-s),\label{4.19}\\
\frac12\chi''&=&r\left([2r+1]\varphi^2 +r[\chi^2+\rho^2]-1\right)\chi,\label{4.20}\\
\frac12\rho''&=&r\left(r[\chi^2+\rho^2]+[\varphi^2-1]+2r[\varphi-s]^2\right)\rho.\label{4.21}
\eea
The superpotential is seen to be \cite{BLW}
\be\label{4.22}
W(\varphi,\chi,\rho)=\varphi-\frac13\varphi^2-r\varphi(\chi^2+\rho^2)+rs\rho^2.
\ee
With $u=(\varphi,\chi,\rho)$, there are exactly six ground states \cite{BLW}:
\be
u^{1,2}=(\pm1,0,0),\quad u^{3,4}=\left(0,\pm\sqrt{\frac1r},0\right),\quad u^{5,6}=\left(s,0,\mp\sqrt{\frac1r(1-s^2)}\right).
\ee

Besides, from (\ref{4.22}), we can read off the BPS equations,
\be\label{4.24}
\left(\varphi',\chi',\rho'\right)=\pm\left(\varphi^2-1+r(\chi^2+\rho^2),2r\varphi\chi,2r(\varphi-s)\rho\right).
\ee

Furthermore, the Hessian of $W$ is
\be
H(u)=\left(\begin{array}{ccc}W_{\varphi\varphi}&W_{\varphi\chi}&W_{\varphi\rho}\\W_{\varphi\chi}&W_{\chi\chi}&W_{\chi\rho}\\W_{\varphi\rho}&W_{\chi\rho}&W_{\rho\rho}\end{array}\right)=-2\left(\begin{array}{ccc}\varphi&r\chi&r\rho\\ r\chi& r\varphi & 0\\ r\rho&0&r(\varphi-s)\end{array}\right),
\ee
which gives us the results
$
H(u^{1,2})=\mp 2\mbox{Diag}\{1,r,r(1\mp s)\},
$
which is definite,
\be
H(u^{3,4})=2\left(\begin{array}{ccc}0&\mp \sqrt{r}&0\\ \mp\sqrt{r}&0&0\\0&0& rs\end{array}\right),
\ee
which is indefinite, and
\be
H(u^{5,6})=-2\left(\begin{array}{ccc}s&0&\mp\sqrt{r(1-s^2)}\\0&rs&0\\\mp\sqrt{r(1-s^2)}&0&0\end{array}\right),
\ee
which is also indefinite. From these results we can
apply Theorem \ref{thm} to immediately conclude that, a domain wall connecting $u^a$ and $u^b$ with
\be
a=1,2,\quad b=1,\dots,6,
\ee
is necessarily BPS. If $a=1,2, b=1,2$, we see from Theorem \ref{thm} that all solutions to (\ref{4.19})--(\ref{4.21})
are contained in the set of solutions of (\ref{4.24}). Hence we may conclude with $\chi\equiv0,\rho\equiv0$ by virtue of
the boundary conditions on $\chi$ and $\rho$.
In other words,  we
see that in this case the only domain wall solutions are given by those of the single scalar field model
\be
{\cal E}=\frac12{\varphi'}^2+\frac12(\varphi^2-1)^2.
\ee
If $a=1,2$ and $b=3,4$,  we apply Theorem \ref{thm} to conclude that all solutions to (\ref{4.19})--(\ref{4.21})
are the solutions to (\ref{4.24}). Thus we infer $\rho\equiv0$ and we see that in this case the
only domain wall solutions are those given by the two-scalar field model (\ref{4.1}) studied earlier. If $a=1,2$ and $b=5,6$,
we again apply Theorem \ref{thm} to see that all solutions to (\ref{4.19})--(\ref{4.21})
are those to (\ref{4.24}). Hence $\chi\equiv0$. Thus we see that in this case the domain wall solutions
are those of the two-scalar field model governed by the energy density
\be
{\cal E}=\frac12{\varphi'}^2+\frac12{\rho'}^2
+\frac12(1-\varphi^2)^2+\frac12r^2 \rho^4+r(\varphi^2-1)
\rho^2+2r^2 (\varphi-s)^2\rho^2.\label{4.30}
\ee
The superpotential of the model (\ref{4.30}) is
\be\label{4.31}
W(\varphi,\rho)=\varphi-\frac13\varphi^3-r(\varphi-s)\rho^2,
\ee
which leads to the BPS equations
\be\label{4.32}
(\varphi',\rho')=\pm\left(1-\varphi^2-r\rho^2,-2r(\varphi-s)\rho\right).
\ee
The ground states in terms of $v=(\varphi,\rho)$ are
\be
v^{1,2}=(\pm1,0),\quad v^{3,4}=\left(s,\pm\sqrt{\frac1r(1-s^2)}\right).
\ee
The Hessian of the superpotential (\ref{4.31}) is
\be
H(v)=-2 \left(\begin{array}{cc}\varphi&r\rho\\r\rho&r(\varphi-s)\end{array}\right)
\ee
so that $H(v)$ is definite for $v=v^{1,2}$ but indefinite for $v=v^{3,4}$. Thus Theorem \ref{thm} ensures that all finite-energy solutions of the equations of motion of (\ref{4.30}), that is,
\bea
\frac12\varphi''&=&(\varphi^2-1)\varphi+r\rho^2\varphi+2r^2\rho^2(\varphi-s),\label{4.35}\\
\frac12\rho''&=&\left(r^2\rho^2+r(\varphi^2-1)+2r^2(\varphi-s)\right)\rho,\label{4.36}
\eea
 linking $v^a$ ($a=1,2$) to $v^b$ ($b=1,\cdots,4$), are given
by those of the BPS equations (\ref{4.32}), the latter have already been obtained in \cite{BLW}. On the other hand, it is
clear that $W(v^3)=W(v^4)$ and
the BPS equations (\ref{4.32}) have no solution linking $v^3$ and $v^4$.  

The equations (\ref{4.35})--(\ref{4.36})
do not permit a further ansatz with $\varphi\equiv s$. At this moment, a domain-wall solution
of  (\ref{4.35})--(\ref{4.36}) linking $v^3$ and $v^4$ remains unknown.

\section{Domain walls in a supersymmetric gauge theory and more examples of non-BPS solutions}
\setcounter{equation}{0}

In \cite{Witten} Witten proposed a product Abelian Higgs model hosting superconducting strings which are relevant to 
cosmology. In \cite{Morris} Morris came up with a supersymmetric extension of the model of Witten \cite{Witten}. In
\cite{Bur} Burinskii showed that the Morris model \cite{Morris} may be adapted to 
give rise to a bag-like superconducting core
in the Kerr--Newmann metric and a BPS domain wall appears to separate the underlying supersymmetric vacuum states.
In this section we illustrate
the limitation as well as applicability of Theorem \ref{thm} in understanding such supersymmetric domain walls.

Recall that in \cite{Bur} the reduced energy density in normalized units, governing the static configuration of a triplet of 
real-valued scalar superfields, $Z,\Sigma, \Phi$, is
\be\label{3.1}
{\cal E}=
\frac12 (Z')^2+\frac12(\Sigma')^2+\frac12(\Phi')^2+V(Z,\Sigma,\Phi),
\ee
where the potential  $V$ assumes the form
\be\label{7}
V(Z,\Sigma,\Phi)=\frac{\lm^2}2\left\{\frac14(\Sigma^2+\Phi^2-\eta^2)^2+Z^2\Sigma^2+\bigg(Z+\frac m\lm\bigg)^2
\Phi^2\right\},
\ee
with $m,\eta,\lm>0$ being the coupling parameters. It may be examined that (\ref{7}) is of the type (\ref{x1.2}), 
generated from a superpotential $W$:
\be
V=\frac12\left(W_Z^2+W_\Sigma^2+W_\Phi^2\right),\quad
 W=\frac\lm2(\eta^2-\Sigma^2-\Phi^2)Z-\frac m2\Phi^2.\label{y5.4}
\ee
The Euler--Lagrange equations of (\ref{3.1}) are
\bea
Z''&=&\lm^2\left(\Sigma^2 Z+\Phi^2\left[Z+\frac m\lm\right]\right),\label{a1}\\
\Sigma''&=&\frac{\lm^2}2(\Sigma^2+\Phi^2-\eta^2)\Sigma+\lm^2 Z^2\Sigma,\label{a2}\\
\Phi''&=&\frac{\lm^2}2(\Sigma^2+\Phi^2-\eta^2)\Phi+\lm^2\left(Z+\frac m\lm\right)^2\Phi.\label{a3}
\eea
In terms of $u=(Z,\Sigma,\Phi)$ we see that there are exactly four domains which are the critical points of $W$:
\be\label{y5.8}
u^1=(0,-\eta,0),\quad u^2=(0,\eta,0),\quad u^3=\left(-\frac m\lm,0,-\eta\right),\quad u^4=\left(-\frac m\lm,0,\eta\right).
\ee

The Hessian matrix of the superpotential $W$ is
\be\label{y5.9}
H(u)=
\left(\begin{array}{ccc}W_{ZZ}&W_{Z\Sigma}&W_{Z\Phi}\\W_{Z\Sigma}&W_{\Sigma\Sigma}&W_{\Sigma\Phi}\\W_{Z\Phi}&W_{\Sigma\Phi}&
W_{\Phi\Phi}\end{array}\right)=-\left(\begin{array}{cccc}0&\lm\Sigma&\lm\Phi\\\lm\Sigma&\lm Z&0\\\lm\Phi&0&\lm Z+m\end{array}\right),
\ee
which can never be definite anywhere.

The BPS system of (\ref{3.1}) reads
\be\label{y5.10}
(Z',\Sigma',\Phi')=\pm(W_Z,W_\Sigma,W_\Phi)=\mp\left(\frac\lm2 (\Sigma^2+\Phi^2-\eta^2),\lm Z\Sigma,
(\lm Z+m)\Phi\right).
\ee
Due to the indefiniteness of (\ref{y5.9}), Theorem \ref{thm} is not applicable to establish the equivalence of the
Euler--Lagrange equations (\ref{a1})--(\ref{a3}) and (\ref{y5.10}). In fact we can find domain wall solutions of
(\ref{a1})--(\ref{a3}) linking certain two domains given in (\ref{y5.8}) which do not satisfy the BPS equations (\ref{y5.10}).

Indeed, it is clear that if we look for a solution  $u=(Z,\Sigma,\Phi)$ of (\ref{a1})--(\ref{a3})  satisfying 
the boundary condition
\be\label{y5.11}
u(\pm\infty)=u^{3,4},
\ee
we may use the ansatz
\be\label{y5.12}
Z\equiv-\frac m\lm,\quad \Sigma\equiv0,\quad \Phi=\mbox{unknown},
\ee
which is not permissible in (\ref{y5.10}).
Then the equations (\ref{a1})--(\ref{a3}), subject to (\ref{y5.11}),  are reduced into the single equation
\be\label{y5.13}
\Phi''=\frac{\lm^2}2(\Phi^2-\eta^2)\Phi,
\ee
subject to $\Phi(\pm\infty)=\pm\eta$, so that the associated energy enjoys the lower bound
\bea
E&=&\int_{-\infty}^\infty {\cal E}\left(-\frac m\lm,0,\Phi\right)\,\dd x=\int_{-\infty}^\infty\left\{\frac12(\Phi')^2+\frac{\lm^2}8
(\Phi^2-\eta^2)^2\right\}\,\dd x\nn\\
&=&\frac12\int_{-\infty}^\infty\left(\Phi'\pm\frac\lm2(\Phi^2-\eta^2)\right)^2\,\dd x\mp\frac\lm 2\int_{-\infty}^\infty
(\Phi^2-\eta^2)\Phi'\,\dd x\nn\\
&\geq&\frac23\lm\eta^3,
\eea
which is attained when $\Phi$ satisfies the BPS equation
\be\label{y5.15}
\Phi'\pm\frac\lm2(\Phi^2-\eta^2)=0.
\ee
Since the superpotential $W(\Phi)=\frac\lm2\left(\frac13\Phi^3-\eta^2\Phi\right)$ of the reduced energy
satisfies $W''(\pm\eta)=\pm\lm\eta\neq0$, we can use Theorem \ref{thm} to infer that (\ref{y5.13}) and (\ref{y5.15})
are equivalent, which is a well known, classic, fact. The equation (\ref{y5.15}) can be easily integrated to give us the solution
\be
\Phi(x)=\pm\eta\tanh\left(\frac{\lm\eta}2(x-x_0)\right).
\ee

However, the BPS system (\ref{y5.10}) subject to the boundary condition (\ref{y5.11})  has no solution under the ansatz
(\ref{y5.12})
as already observed. In other words we have again obtained some non-BPS solutions of (\ref{a1})--(\ref{a3})  subject to (\ref{y5.11}).

Similarly we can show that (\ref{a1})--(\ref{a3}) subject to the boundary condition
\be
u=(\pm\infty)=u^{1,2}
\ee
allow solutions with the ansatz
\be
Z\equiv 0,\quad \Sigma=\mbox{unknown},\quad \Phi\equiv0,
\ee
but the system (\ref{y5.10}) does not allow such solutions. Thus the existence of non-BPS solutions again follows.

Note that for the superpotential $W(u)$ defined in (\ref{y5.4}), we have $W(u^1)=W(u^2)=0$ and $W(u^3)=W(u^4)=
-\frac m2\eta^2$. Hence the BPS bound vanishes for the field configurations linking both $u^1$ and $u^2$, and, $u^3$ and $u^4$, although the energy of the solutions of (\ref{a1})--(\ref{a3}) obtained above for both the
cases is
\be
E=\frac23\lm\eta^3,
\ee
which is strictly positive. Thus the BPS bound in these cases is not attainable and there are domain walls of energy above
the BPS bound.

\section{Existence of supersymmetric BPS domain walls}
\setcounter{equation}{0}

On the other hand, the BPS bound in the Burinskii--Morris model \cite{Bur,Morris} for the field configurations linking $u^{1,2}$,
referred to in \cite{Bur} as the supersymmetric vacuum state I,  and $u^{3,4}$, referred to in \cite{Bur} 
as supersymmetric
vacuum state II, is positive:
\be
\left|W(u^{1,2})-W(u^{3,4})\right|=\frac12 m\eta^2,
\ee
which prompts the question whether the BPS system (\ref{y5.10}) permits solutions, which will be studied in this section.
 Since (\ref{y5.10}) cannot be integrated, we construct its solutions by analytic means.

Without loss of generality, we consider the upper sign situation of the system (\ref{y5.10}) subject to the boundary condition
\be\label{ubc}
u(-\infty)=u^4,\quad u(\infty)=u^2.
\ee
 For convenience, using the new variables, 
\be\label{8}
Z=\frac m\lm h,\quad \Sigma =\eta f,\quad \Phi=\eta g,
\ee
we obtain the rescaled equations and the corresponding boundary condition,
\be\label{10}
\left(mh',f',g'\right)=-\left(\beta (f^2-1)+\beta g^2, mfh,m(h+1)g\right),
\ee
and
\be\label{11}
(h,f,g)(-\infty)=(-1,0,1),\quad (h,f,g)(\infty)=(0,1,0).
\ee
where 
$
\beta=\frac12\lm^2\eta^2.
$

Applying the uniqueness theorem for the initial value problem of ordinary differential
equations, we see that the function $f$ and $g$ can never vanish. Hence, we may assume
$f(x)>0$ and $g(x)>0$ for all $x\in(-\infty,\infty)$. In view of this fact and the
system (\ref{10}), we have
\be\label{16}
\ln g(x)=\ln f(x)-mx +k,
\ee
where $k$ is an arbitrary constant, which can be translated away by the independent variable
$x$. Hence, we may well assume $k=0$ in (\ref{16}), which results in the relation
\be\label{17}
g(x)= f(x)\e^{-mx}.
\ee

Inserting the second relation, $f'=-mfh$, in (\ref{10}) and (\ref{17}) into the first relation in (\ref{10}), we arrive at the scalar equation
\be\label{18}
(\ln f)''=\beta (f^2-1)+\beta f^2 \e^{-2m x}.
\ee

For convenience, we may introduce the new variable 
\be\label{19}
u=2\ln f.
\ee
Therefore (\ref{18}) becomes a Liouville type equation in dimension one,
\be\label{20}
u''=\Lambda (\e^u-1)+\Lambda \e^{-2m x+u},
\ee
where
\be\label{21}
\Lambda=2\beta=\lm^2\eta^2.
\ee

In view of the boundary condition (\ref{11}) and the relation (\ref{19}), we see that we have the following boundary condition for $u$:
\be\label{22}
u(-\infty)=-\infty,\quad u(\infty)=0.
\ee

Thus, we are to solve the two-point boundary value problem consisting of (\ref{20}) and
(\ref{22}) over the full interval $(-\infty,\infty)$. It is clear that any solution $u$ of (\ref{20}) satisfying (\ref{22}) must be
negative-valued. In fact, using (\ref{22}), we see that if $u\geq0$ somewhere then there exists a point $x_0\in\bfR$
where $u$ attains its global maximum in $\bfR$. Thus $u''(x_0)\leq0$ which leads to a contradiction in view of (\ref{20}).

In order to solve this problem, we shall use a dynamical shooting technique. For this purpose,
we consider, instead, the initial value problem:
\bea
u''&=&\Lambda (\e^u-1)+\Lambda \e^{-2m x +u},\quad -\infty<x<\infty,\label{23}\\
u(0)&=&-a, \quad u'(0)=b.\label{24}
\eea

Recall that we are interested in negative solutions. So we assume $a>0$ in (\ref{24}).

\begin{proposition}\label{Proposition1}
For any $a>0$, there is a unique $b>0$, so that the solution of (\ref{23}) and (\ref{24})
 satisfies $u(\infty)=0$.
\end{proposition}

In order to prove this proposition, we define
\bea
{\cal B}^-&=&\{b\in\bfR\,|\, \exists x>0 \mbox{ so that }u'(x)<0\},\nn\\
{\cal B}^0&=&\{b\in\bfR\,|\, u'(x)>0\mbox{ and } u(x)\leq0\mbox{ for all }x\geq0\},\nn\\
{\cal B}^+&=&\{b\in\bfR\,|\,u'(x)>0\mbox{ for all } x\geq0\mbox{ and } u(x)>0\mbox{ for some }
x>0\}.\nn
\eea

Note that in the above definitions of the sets $\cal B$'s, the solution
$u(x)$ need not exist for all $x$ and 
the statements should actually be read as for all $x$ where $u(x)$ exists.

\begin{lemma}\label{Lemma1}
The real line $\bfR$ is the disjoint union of the sets ${\cal B}^-,{\cal B}^0,{\cal B}^+$.
In particular, if $b\not\in{\cal B}^-$, then $u'(x)>0$ for all $x>0$ in the interval of
existence of the solution.
\end{lemma}
\begin{proof} 
If $b\not\in {\cal B}^-$, then $u'(x)\geq0$ for all $x>0$ in the interval of existence of
the solution $u$. If there is a point $x_0>0$ in the
interval of existence of $u$ so that $u'(x_0)=0$, then $u''(x_0)=0$ because otherwise
$u'(x)<0$ for $x$ close to $x_0$ but $x<x_0$ if $u''(x_0)>0$ or $x>x_0$ if
$u''(x_0)<0$ which contradicts the assumption that $b\not\in{\cal B}^-$.

Besides, differentiating (\ref{23}), we have
\bea\label{25}
u'''(x_0)&=&\Lambda \e^{u(x_0)}u'(x_0)+\Lambda \e^{-2m x_0 +u(x_0)}\left(-2m +u'(x_0)\right)\nn\\
&=&-2m\Lambda \e^{-2m x_0
+u(x_0)}<0.
\eea
Hence, using $u''(x_0)=0$ and (\ref{25}), we see
that $u''(x)<0$ for $x>x_0$ but $z$ is close to $x_0$. Combining this fact with the
assumption $u'(x_0)=0$, we derive $u'(x)<0$ for $x>x_0$ but $x$ is close to $x_0$, which
again contradicts $b\not\in{\cal B}^-$.

In other words, we have shown that $b\not\in{\cal B}^-$ implies $u'(x)>0$ for all $x$.
Therefore, $b\in {\cal B}^0\cup{\cal B}^+$.
\end{proof}

\begin{lemma}\label{Lemma2}
The sets ${\cal B}^-$ and ${\cal B}^+$ are both open and nonempty.
\end{lemma}
\begin{proof}
The statement for ${\cal B}^-$ is obvious since $(-\infty,0)\subset{\cal B}^-$. For the
statement concerning ${\cal B}^+$, we integrate (\ref{23}) to get
\bea
u'(x)&=&b+\Lambda \int_0^x\left\{(\e^{u(s)}-1)+\e^{-2ms+u(s)}\right\}\,\dd s,\label{26}\\
u(x)&=&-a+bx+\Lambda \int_0^x\int_0^t\left\{(\e^{u(s)}-1)+\e^{-2ms+u(s)}\right\}\,\dd s\dd t.\label{27}
\eea
For any $x_0$, we can choose $b>0$ sufficiently large so that
\bea
b-\Lambda x_0>0,\label{28}\\
-a+bx_0-\frac12\Lambda x_0^2>0.\label{29}
\eea
We prove that if $b$ satisfies (\ref{28}) and (\ref{29}), then $b\in
{\cal B}^+$. To this end, let $[0,X_0)$ be the 
(forward) interval of existence of the solution of
(\ref{23}) and (\ref{24}). If $X_0\leq x_0$, then (\ref{26}) and (\ref{28}) give us
$
u'(x)>b-\Lambda x>b-\Lambda x_0>0$ for $ 0\leq x<X_0.
$
Since $X_0<\infty$, there must be a point $x_1\in (0,X_0)$ so that $u(x_1)>0$, which proves
$b\in{\cal B}^+$. If $X_0>x_0$, then (\ref{26})--(\ref{29}) give us
\bea
u'(x)&>&b-\Lambda x_0>0,\quad 0\leq x\leq x_0,\label{30}\\
u(x_0)&>&-a +bx_0-\frac12\Lambda x_0^2>0.\label{31}
\eea
Using (\ref{26}) and (\ref{30}), (\ref{31}), we see that $u'(x)>0$ and $u(x)>0$ for all
$x_0\leq x<X_0$. In particular, $b\in{\cal B}^+$. Thus we have established the fact that
${\cal B}^+\neq\emptyset$.

To see that ${\cal B}^+$ is open, let $b_0\in {\cal B}^+$ and $u(x;b_0)$ be the corresponding
solution of (\ref{23}) and (\ref{24}) so that $u(x_0;b_0)>0$ for some $x_0>0$. By the
continuous dependence theorem for the solutions of the initial value problems of
ordinary differential equations, we can find a neighborhood of $b_0$, say
$(b_0-\vep,b_0+\vep)$ ($\vep>0$), so that for any $b\in (b_0-\vep,b_0+\vep)$, the
solution of (\ref{23}) and (\ref{24}), say $u(x;b)$, satisfies $u'(x; b)>0$ for
$0\leq x\leq x_0$ and $u(x_0;b)>0$. Since $u'(x;b)$ satisfies
\be
u'(x;b)=u'(x_0;b)+\Lambda\int_{x_0}^x\left\{(\e^{u(s;b)}-1)+\e^{-2ms+u(s;b)}\right\}\,\dd s,\quad x\geq x_0,
\ee
we see that $u'(x)>0$ and $u(x;b)>u(x_0;b)>0$ for all $x>x_0$. This proves $b\in{\cal B}^+$
and the openness of ${\cal B}^+$ follows.
\end{proof}

\begin{lemma}\label{Lemma3}
The set ${\cal B}^0$ is nonempty and closed. Furthermore, if $b\in{\cal B}^0$, then
the solution of (\ref{23}) and (\ref{24}) 
satisfies $u(x)<0$ for all $x>0$. In particular, the solution exists for all $x>0$.
\end{lemma}
\begin{proof}
The first part is a consequence of the connectedness of $\bfR$. To prove the second part,
we assume that there is a $x_0>0$ such that $u(x_0)=0$. Such a point $x_0$ is a local
maximum point of the function $u(x)$ because $u(x)\leq0$ for all $x\geq0$. However, inserting
this fact to (\ref{23}), we have $u''(x_0)=\Lambda\e^{-2mx_0}>0$, which is a contradiction.
\end{proof}

\begin{lemma}\label{Lemma4}
For $b\in {\cal B}^0$, we have
\be\label{32}
\lim_{x\to\infty} u(x)=0.
\ee
\end{lemma}
\begin{proof}
Since $u'(x)>0$ and $u(x)<0$ for all $x\geq0$, we see that the limit
$
\lim_{x\to\infty}u(x)=u_0
$
exists and $-\infty <u_0\leq0$. If $u_0<0$, then $u(x)<u_0$ for all $x\geq0$. Using (\ref{23}),
we can find a sufficiently large $x_0>0$ such that
$
u''(x)<\frac12\Lambda (\e^{u_0}-1)\equiv -\delta, x\geq x_0.
$
In particular, $u'(x)<u'(x_0)-\delta(x-x_0)$ ($x>x_0$), which leads to a contradiction when $x$ is sufficiently large because $u'(x)>0$ for all $x>0$.
\end{proof}

\begin{lemma}\label{Lemma5}
The set ${\cal B}^0$ is a single point.
\end{lemma}
\begin{proof}
Let $u_1(x)$ and $u_2(x)$ be the two solutions of (\ref{23}) and (\ref{24}) when
$b=b_1$ and $b=b_2$, respectively, where $b_1$ and $b_2$ are taken from ${\cal B}^0$.
Then $w(x)=u_1(x)-u_2(x)$ satisfies
\bea
w''(x)&=&\Lambda\left(\e^{\xi(x)}+\e^{-2mx+\xi(x)}\right) w(x), \quad 0<x<\infty,\label{33}\\
w(0)&=&0,\quad w(\infty)=0,\label{34}
\eea
where $\xi(x)$ lies between $u_1(x)$ and $u_2(x)$. Applying the maximum principle to the above,
we find $w(x)=0$ everywhere. In particular, $b_1=b_2$.
\end{proof}

\begin{lemma}\label{Lemma6}
For $b\in {\cal B}^0$, the solution $u(x)$ of (\ref{23}) and (\ref{24}) enjoys the
sharp asymptotic estimates for $x$ near $\infty$:
\bea
0&>&u(x)\geq -C(\vep)\e^{-\min\{\sqrt{\Lambda},2m\}(1-\vep)x},\label{35}\\
0&<&u'(x)\leq C(\vep)\e^{-\min\{\sqrt{\Lambda},2m\}(1-\vep)x},\label{36}
\eea
where $\vep>0$ is an arbitrarily small number and $C(\vep)>0$ is a constant.
\end{lemma}
\begin{proof}
We can use the mean-value theorem to rewrite (\ref{23}) as
\be\label{37}
u''(x)=\Lambda \e^{\xi(x)}u(x)+\Lambda \e^{-2mx+u(x)},
\ee
where $\xi(x)\in (u(x),0)$. Since $u(x)$ goes to zero as $x\to\infty$, we may view the
second term on the right-hand side of (\ref{37}) as a source term which vanishes at
$x=\infty$ in the order O($\e^{-2mx}$). Therefore, the estimate (\ref{35}) follows
from (\ref{37}) because $\e^{\xi(x)}\to 1$ as $x\to\infty$.

Inserting (\ref{35}) into (\ref{37}), we see that $u''(x)$ satisfies
$|u''(x)|=\mbox{O}(\e^{-\min\{\sqrt{\Lambda},2m\}(1-\vep)x})$ for any small $\vep>0$. Since
$u'(\infty)=0$, we see that (\ref{36}) follows as well.
\end{proof}

Since for each $a>0$, the set ${\cal B}^0$ contains exactly one point, $b>0$, we may denote
this point as $b_a$. Namely,
\be\label{38}
{\cal B}^0=\{b_a\}.
\ee
Our discussion above has shown that ${\cal B}^-=(-\infty,b_a)$ and
${\cal B}^+=(b_a,\infty)$. This knowledge will be useful in numerical solution of the
problem. The following behavior of $b_a$ with regard to $a$ is important for our construction
at the other end of the interval, $x=-\infty$.

\begin{lemma}\label{Lemma7}
The number $b_a$ defined in (\ref{38}) has the property
\be\label{39}
\lim_{a\to 0^+}b_a=0.
\ee
\end{lemma}
\begin{proof} Let $u$ be the solution of (\ref{23}) and (\ref{24}) where $b=b_a$.
Since $u'(x)>0$ for $x>0$, we have
\be\label{40}
\int_0^{\infty} \e^{-2mx+u(x)}\dd x>\frac{\e^{-a}}{2m}.
\ee
Multiplying (\ref{23}) by $u'$, integrating over $(0,\infty)$, and using $u'(\infty)=0$,
we have 
\bea
b_a^2&=&2\Lambda (\e^{-a}+a-1)+2\Lambda\e^{-a}-4m\Lambda \int_0^{\infty}\e^{-2mx+u(x)}\dd x\nn\\
&<&2\Lambda (\e^{-a}+a-1)+2\Lambda\e^{-a}-2\Lambda \e^{-a},\label{41}
\eea
where we have used (\ref{40}). Letting $a\to 0^+$ in (\ref{41}), we see that the proof follows.
\end{proof}

Next, we consider the solution in the interval $(-\infty,0)$. Note that the boundary condition
$u(-\infty)=-\infty$ is not sufficient for us to recover the boundary asymptotics specified
for $Z$ and $\Phi$ in (\ref{ubc})  or $h$ and $g$ in (\ref{11}).
For example, in view of (\ref{17}) and (\ref{19}), we have $g(x)=\e^{-mx+\frac12 u(x)}$. Therefore, this relation and
the last condition in (\ref{11}) imply that we need to achieve the precise condition
\be\label{42}
\lim_{x\to-\infty}(-2mx +u(x))=0.
\ee
Besides, combining the equation (\ref{10}) and relation (\ref{19}), we see that the first condition in (\ref{11})
requires us to get
\be\label{43}
\lim_{x\to-\infty} u'(x)= 2m.
\ee
These two compatible conditions seem to be more subtle to realize. In the subsequent analysis,
we solve this problem.

For convenience, we use the new variable $v=-2mx+u$ and change $x$ to $-x$. Hence, for
$-\infty<x<0$, (\ref{23}) and (\ref{24}) give us a new initial value problem in terms of
$v$ and $x$ as follows,
\bea
v''&=& \Lambda \e^{-2m x +v}+\Lambda (\e^v-1),\quad x>0,\label{44} \\
v(0)&=&-a,\quad v'(0)=2m-b_a.\label{45}
\eea
It is interesting that (\ref{44}) is of the same form as (\ref{23}) for $x>0$. The difference
here is that in the initial condition (\ref{45}), the parameter $a>0$ is to be adjusted
as a shooting parameter, but not the slope $v'(0)$ which appears indirectly and depends on
$a$.

Now define the disjoint sets
\bea
{\cal A}^-&=&\{a>0\,|\,\mbox{There exists }x>0 \mbox{ such that } v'(x)<0\},\nn\\
{\cal A}^0&=&\{a>0\,|\, v'(x)>0\mbox{ and } v(x)\leq 0\mbox{ for all }x\geq 0\},\nn\\
{\cal A}^+&=&\{a>0\,|\, v'(x)>0\mbox{ for all }x\geq0\mbox{ and there is an }x>0
\mbox{ such that } v(x)>0\}.\nn
\eea
Note that, as before, the above statements should be understood to mean in the interval of existence of the concerned solutions.

\begin{lemma}\label{Lemma8}
The set ${\cal A}^+$ is nonempty and open.
\end{lemma}
\begin{proof} The proof is divided into a few steps.

{\em Step 1.} Let $a\in{\cal A}^+$ and $x_0>0$ in the interval of existence
of the solution $v$ so that $v(x_0)>0$. Then $v'(x)>0$ for all $x\in[0,x_0]$.

In fact, if there is an $x_1\in (0,x_0]$ so that $v'(x_1)=0$, then $v''(x_1)=0$
otherwise we would have $v'(x)<0$ for $x>x_1$ ($x<x_1$) but $x$ is near $x_1$ when
$v''(x_1)<0$ ($v''(x_1)>0$), which contradicts the condition that $a\in{\cal A}^+$.
Hence, differentiating (\ref{44}), we have
$
v'''(x_1)=-2m\Lambda \e^{-2m x_1+v(x_1)}<0.
$
In particular, $v'(x)$ is concave down in a neighborhood of $x_1$ and the properpty
that $v'(x_1)=0$ implies that $v'(x)<0$ for $x\neq x_1$ but $x$ is near $x_1$, which is
another contradiction.

{\em Step 2.} The set ${\cal A}^+$ is nonempty.

To prove the claim, we consider the initial value problem
\bea
w''&=& \Lambda \e^{-2m x +w}+\Lambda (\e^w-1),\quad x>0,\label{46} \\
w(0)&=&0,\quad w'(0)=2m.\label{47}
\eea
Of course, the solution $w$ of (\ref{46}) and (\ref{47}) exists over an interval $[0,x_0]$
for some $x_0>0$ and it is clear that $w(x)>0$ for $0<x\leq x_0$ and $w'(x)>2m$ for
$0\leq x\leq x_0$. By the continuous dependence theorem, we see that, when $a>0$ is small
enough, the unique solution $v$ of (\ref{44}) and (\ref{45}) satisfies
$v'(x)>m>0$ for $0\leq x\leq x_0$ and $v(x_0)>0$. Now for $x>x_0$, using (\ref{44})
and $v'(x_0)>0$, we see that $v'(x)>0$ for all $x>x_0$ as well. This proves $a\in{\cal A}^+$.

{\em Step 3.} The set ${\cal A}^+$ is open.

In fact, let $a_0\in {\cal A}^+$. Then, by Step 1, the solution $w$ of (\ref{44}) and (\ref{45})
with $a=a_0$ satisfies $w(x_0)>0$ and $w'(x)>0$ for all $x\in[0,x_0]$ for some $x_0>0$. Using
the same continuity argument as in Step 2, we deduce that the solution $v$ of
(\ref{44}) and (\ref{45}) satisfies $v'(x)>0$ for all $x>0$ and $v(x_0)>0$ for the parameter
$a$ near $a_0$. Hence the openness of ${\cal A}^+$ follows.
\end{proof}

\begin{lemma}\label{Lemma9}
The set ${\cal A}^-$ is open and nonempty.
\end{lemma}
\begin{proof}
The openness of ${\cal A}$ is obvious. To see that ${\cal A}$ is nonempty, we return to
the solution $u$ of (\ref{23}) and (\ref{24}) and use the notation of Lemma \ref{Lemma7}.
Inserting
$
\int_0^{\infty} \e^{-2m x+u(x)}\dd x<\frac1{2m}
$
into the first line in (\ref{41}), we have
$
b_a^2>2\Lambda(2\e^{-a}+a-1)-2\Lambda.
$
In particular, $b_a>2m$ when $a>0$ is sufficiently large. This establishes $a\in{\cal A}^-$
in view of (\ref{45}).
\end{proof}

\begin{lemma}\label{Lemma10}
The interval $(0,\infty)$ is the disjoint union of the sets ${\cal A}^-,{\cal A}^0,
{\cal A}^+$. In particular, ${\cal A}^0\neq\emptyset$.
\end{lemma}
\begin{proof}
The first statement has already been established. The second statement is a consequence of
the connectedness of the interval which cannot be the union of two nonoverlapping and
nonempty open sets.
\end{proof}

\begin{lemma}\label{Lemma11}
Let $a\in {\cal A}^0$ and $v$ be the corresponding solution of (\ref{44}). Then
the solution exists globally and $v(x)<0$ for all $x>0$. Besides, $v(x)\to0$ as
$x\to\infty$ and there hold the asymptotic estimates near $x=\infty$:
\bea
0&>&v(x)\geq -C(\vep)\e^{-\min\{\sqrt{\Lambda},2m\}(1-\vep)x},\label{48}\\
0&<&v'(x)\leq C(\vep)\e^{-\min\{\sqrt{\Lambda},2m\}(1-\vep)x},\label{49}
\eea
where $\vep>0$ is an arbitrarily small number.
\end{lemma}
\begin{proof}
See the proofs for Lemmas \ref{Lemma3}, \ref{Lemma4}, and \ref{Lemma6}.
\end{proof}

Now returning to the original variables $-x\mapsto x$ and $u(x)=v(-x)+2m x$, we see that
we have obtained a suitable number $a>0$ (with the unique corresponding $b=b_a$) so 
that the solution $u$ to the initial value problem (\ref{23}) and (\ref{24}) exists
over the entire interval $-\infty<x<\infty$ and enjoys the properties stated in
Lemmas \ref{Lemma3} and \ref{Lemma6} for $x>0$ and $u(x)-2mx\to0$ as $x\to-\infty$. 
Moreover, for $x<0$,
$u(x)-2m x<0$ and $u'(x)-2m<0$, and for $x$ near $-\infty$ we have the asymptotic estimates
\bea
0&>&u(x)-2m x\geq -C(\vep)\e^{\min\{\sqrt{\Lambda},2m\}(1-\vep)x},\label{50}\\
0&>&u'(x)-2m\geq -C(\vep)\e^{\min\{\sqrt{\Lambda},2m\}(1-\vep)x},\label{51}
\eea
where $\vep>0$ can be chosen to be arbitrarily small.

\begin{lemma}\label{Lemma12}
The suitable number $a>0$ stated above is also unique.
\end{lemma}
\begin{proof} Suppose that $a_1>0$ and $a_2>0$ are two suitable numbers defined above
and $u_1(x)$ and $u_2(x)$ are the corresponding solutions, respectively. Then
$w=u_1-u_2$ satisfies $w(-\infty)=w(\infty)=0$ and the equation (\ref{33}) over
$-\infty<x<\infty$. Using the maximum principle again, we have $w(x)\equiv0$. In particular,
$a_1=a_2$.
\end{proof}

Thus, regarding the existence of domain wall solutions, we can draw the following conclusions.
\begin{enumerate}
\item[(i)] The rescaled BPS system (\ref{10}) subject to the
boundary condition (\ref{11}) has a unique solution $(h,f,g)$ up to translations. 

\item[(ii)] The unique solution obtained enjoys
the sharp exponential asymptotic estimates
\bea
&&-1<h(x)<-1+C(\vep)\e^{\min\{\sqrt{\Lambda},2m\}(1-\vep)x},\nn\\
&&\e^{mx}>f(x)>
\e^{mx}\left(1-C(\vep)\e^{\min\{\sqrt{\Lambda},2m\}(1-\vep)x}\right),\nn\\
&&1>g(x)>1-C(\vep)\e^{\min\{\sqrt{\Lambda},2m\}(1-\vep)x}\quad \mbox{as }x\to-\infty;\nn\\
&&0>h(x)>-C(\vep)\e^{-\min\{\sqrt{\Lambda},2m\}(1-\vep)x},\nn\\
&&1>f(x)>1-C(\vep)\e^{-\min\{\sqrt{\Lambda},2m\}(1-\vep)x},\nn\\
&&\e^{-mx}>g(x)>\e^{-mx}\left(1-C(\vep)\e^{-\min\{\sqrt{\Lambda},2m\}(1-\vep)x}\right)\quad
\mbox{as }x\to\infty,\nn
\eea
where $\vep$ is an arbitrarily small number and $C(\vep)>0$ is a constant depending
on $\vep$.
\end{enumerate}

Returning to the original variables, $Z,\Sigma,\Phi$, and omitting the arbitrarily small parameter $\vep$, we see that we have obtained the existence and uniqueness (modulo
translations) of a domain wall soliton of the BPS system (\ref{y5.10}) (with the upper sign) which fulfills the exact boundary conditions
\bea
Z(x)+\frac m\lm&=&\mbox{O}\left(\e^{\min\{\lm\eta,2m\}x}\right),\quad \Sigma(x)=\mbox{O}\left(\e^{mx}\right),\quad
\Phi(x)-\eta=\mbox{O}\left(\e^{\min\{\lm\eta,2m\}x}\right),\nn\\
 x&\to&-\infty;\\
Z(x)&=&\Sigma(x)-\eta=\mbox{O}\left(\e^{-\min\{\lm\eta,2m\}x}\right),\quad \Phi(x)=\mbox{O}\left(\e^{-mx}\right),\nn\\
x&\to&\infty,
\eea
realizing a phase transition from $u^4=\left(-\frac m\lm,0,\eta\right)$ to $u^2=(0,\eta,0)$ as anticipated.

\section{Conclusions and remarks}

The finite-energy static solutions of the equations of motion of quantum field theory models generally known as
topological solitons are difficult to obtain due to the complicated structure of the equations and their BPS systems
present a  significant reduction and offer great sight into the problems. However, a well-known puzzle is whether
at the BPS critical coupling all such solitons must be BPS, hence achieving the associated minimum BPS energy bounds.
The earlier studies on the puzzle established two main categories of results: In the Abelian Higgs vortex model all finite-energy
solitons in the BPS coupling are BPS \cite{JT,T1} but in the Yang--Mills--Higgs monopole model \cite{T2}
and Yang--Mills instanton
model \cite{Bor,Parker,SS1,SS2,SSU} there exist non-BPS solitons which are of course of energies above the BPS bounds.
In this paper we have carried out a systematic investigation on the puzzle for a general domain wall model
governing a multiple real-component scalar field $u=(u_1,\dots,u_n)$ in terms of a superpotential 
$W(u)$ so that the potential  $V(u)$ is given by $V=\frac12|\nabla W|^2$ and the ground states of the model are
the critical points of $W$. Our work shows that the answer whether or not all finite-energy solitons are BPS essentially lies in the definiteness of the Hessian $H(u)=( W_{u_i u_j}(u))$ of $W$ which leads us to draw the following conclusions.
\begin{enumerate}
\item[(i)] Let $u^a$ and $u^b$ be two ground states which are the critical points of the superpotential $W$. If
the Hessian of $W$ is positive or negative definite at either $u^a$ or $u^b$ then any finite-energy
domain wall solution linking $u^a$ and $u^b$ is necessarily BPS.

\item[(ii)] If the Hessian of $W$ fails to be definite at both $u^a$ and $u^b$ there are examples showing that
 the BPS bounds are not attainable and there exist finite-energy domain wall solutions linking $u^a$ and $u^b$ which are not BPS.
\end{enumerate}

Finally we note that some of the interesting domain wall models are not of the superpotential type studied here but it may be 
shown that all the finite-energy domain wall solitons there are necessarily BPS. These models include those arising in the
monopole confinement problem \cite{ABS,CLY1} and the Skyrme crystal problem \cite{C,CLY2}.

\medskip
\medskip

The authors were partially
supported by National Natural Science Foundation of China under Grant 11471100.

\small{

}

\end{document}